\documentclass[a4paper,12pt,reqno]{article}
\usepackage{amsmath}%
\usepackage{amsfonts} %
\usepackage{amssymb}%
\usepackage{graphicx}%
\usepackage{mathrsfs}%
\usepackage{amsthm}
\usepackage{bbm}

\usepackage{mathtools}  
\mathtoolsset{showonlyrefs}  

\advance \voffset by -0.8in
\advance \hoffset by -0.9in
\theoremstyle{plain}

\newtheorem{theorem}{Theorem}

\newtheorem{assumption}[theorem]{Assumption}

\newtheorem{corollary}[theorem]{Corollary}

\newtheorem{definition}[theorem]{Definition}

\newtheorem{lemma}[theorem]{Lemma}

\newtheorem{remark}[theorem]{Remark}

\usepackage{subcaption}

\usepackage{float}
\restylefloat{table}

\newcommand{\MAH}{M_2^0}
\newcommand{\bee}{\begin{equation}}
\newcommand{\eee}{\end{equation}}
\newcommand{\atn}{a_{\text{\tiny TN}}}
\newcommand{\btn}{b_{\text{\tiny TN}}}
\newcommand{\ctn}{c_{\text{\tiny TN}}}
\newcommand{\RR}{\mathbb{R}}
\newcommand{\ZZ}{\mathbb{Z}}
\newcommand{\CC}{\mathbb{C}}

\newcommand{\cmin} {c_{\text{\tiny min}}}
\newcommand{\VYMH}{V_{\text{\tiny YMH}}}
\newcommand{\HYMH}{H_{\text{\tiny YMH}}}
\newcommand{\tHYMH}{\tilde{H}_{\text{\tiny YMH}}}
\newcommand{\U}{\xi}

\usepackage{fourier}
\usepackage{array}
\usepackage{booktabs}
\usepackage{tabularx}
\usepackage{makecell}

\newcolumntype{Z}{ >{\centering\arraybackslash}X }

\usepackage[autolanguage, np]{numprint}

\begin{document} 

\vskip 10pt
\baselineskip 28pt

\begin{center}
{\Large \bf Quantum Bound States in Yang-Mills-Higgs Theory }

\baselineskip 18pt

\vspace{0.4 cm}

{\bf Lyonell Boulton, Bernd J~Schroers and Kim Smedley-Williams}\\
\vspace{0.2 cm}
Maxwell Institute for Mathematical Sciences and
Department of Mathematics,
\\Heriot-Watt University,
Edinburgh EH14 4AS, UK. \\
{\tt l.boulton@hw.ac.uk}, \;\;  {\tt b.j.schroers@hw.ac.uk}  and {\tt ks275@hw.ac.uk}

\vspace{0.4cm}

{20  August  2018 }

\end{center}

\begin{abstract}
\noindent We give rigorous proofs of  the existence of infinitely many (non-BPS) bound states for two linear operators associated with the Yang-Mills-Higgs equations at vanishing Higgs self-coupling and  for gauge group $SU(2)$: the operator obtained by linearising the Yang-Mills-Higgs equations around a charge one monopole and the Laplace operator on the Atiyah-Hitchin moduli space of centred charge two monopoles. For the linearised system we use the  Riesz-Galerkin approximation to compute upper bounds on the lowest 20 eigenvalues. We discuss the similarities in the spectrum of the linearised  system and the Laplace operator, and interpret them in the light of electric-magnetic duality conjectures.
\end{abstract}

\baselineskip 14pt
\parskip 1 pt
\parindent 10pt

\section{Introduction}
Developments in mathematical physics over the past decades have amply demonstrated that, in order to fully understand a quantum field theory, one also needs to study the corresponding classical field theory and its solutions. The simplest case, where one quantises and studies fluctuations about the vacuum,  turns out to be a rather special one; in general there are other sectors where the vacuum is replaced by non-trivial classical field configurations, often characterised by the  non-vanishing of a topological invariant. 

This is particularly true for gauge theories like Yang-Mills and Yang-Mills-Higgs (YMH) theories. In the latter, the non-trivial classical field configurations are non-abelian monopoles. They satisfy  the non-linear YMH field equations and carry non-vanishing winding numbers  which physically manifest themselves as magnetic charges. 

In this paper we are interested in   two spectral problems which arise in the study of fluctuations around  monopole fields but which, {\em a priori}, are not related in any obvious way.  

One of them is obtained by linearising the YMH equations around a static monopole solution. The second has its origin in a more subtle and intricate feature of the YMH system in the so-called BPS limit, where the Higgs self-coupling vanishes. In that limit there exists a whole manifold of solutions, called the moduli space, which inherits a Riemannian metric from the YMH kinetic energy functional \cite{AHbook}. Associated to this metric there are  naturally defined linear operators like the Laplace-Beltrami, Laplace-de Rham and Dirac operators on the moduli space, with interesting spectral properties.

Examples of both sorts of problem have been studied in the  physics literature, for example in the papers \cite{BT,RS,FV} for the YMH system linearised about a single monopole and in \cite{GM,M,S} for the Laplace-Beltrami operator on the moduli space of two centred monopoles, also called the Atiyah-Hitchin manifold. A combination of analytical and numerical methods has revealed a host of interesting spectral phenomena, including infinitely many Coulomb-like bound states and Fesh\-bach resonances, but there are very few mathematically rigorous results in this respect. 
 
The primary purpose of this paper is to begin to fill this gap by applying techniques of spectral analysis to the linearised YMH equation and the Atiyah-Hitchin Laplacian. We prove the existence of infinitely many bound states in both systems and give upper bounds on the eigenvalues for the linearised YMH equations. In this way we also hope to introduce two interesting problems to the spectral analysis community and, conversely, useful analytical techniques to the community of theoretical physicists interested in magnetic monopoles. 
 
A further motivation of this paper is to exhibit the striking similarity in the spectra of two operators which superficially look very different: the linearised YMH equation is defined on Euclidean three-space while the Atiyah-Hitchin Laplacian is defined on a non-compact Riemannian four-manifold.
The similarities in their spectra are probably related to electric-magnetic duality conjectures in YMH quantum field theory, but the details are far from clear. We end this extended introduction with a summary of relevant background on duality conjectures.  

YMH theory involves an $SU(2)$ gauge field and a Higgs field in the adjoint representation. The symmetry is spontaneously broken to $U(1)$ either via a boundary condition on the Higgs field (in the BPS limit) or by a Higgs potential (at generic coupling). In the BPS limit and in suitable units, the perturbative particle spectrum after symmetry breaking consists of a photon, a massless Higgs scalar and massive W-bosons with equal and opposite  electric charges $n=\pm 1$.  

In addition to the perturbative particles, the theory contains solitonic magnetic monopoles \cite{tHooft, Polyakov}, labelled by an integer magnetic charge $m$. When allowed to evolve in time, classical magnetic monopoles may acquire electric charge, thus becoming dyons (particles with both magnetic and electric charge). After quantisation, the electric charge is characterised by another integer $n$ so that states in quantum YMH theory fall into different sectors labelled by a pair of integer charges $(m, n)$. A magnetic monopole belongs to the sector $(1,0)$, a W-boson to the sector $(0,1)$ or $(0,-1)$, the simplest dyon to the sector $(1,1)$ and so on. 
 
Electric-magnetic duality conjectures relate the properties of sectors with different magnetic and electric charges. In the simplest version, due to Montonen and Olive \cite{MO}, a sector with label $(m,n)$ is conjectured to be equivalent, in a suitable sense, to the sector $(n,-m)$, with electric and magnetic charge being exchanged. 
In the more general S-duality conjecture \cite{Sen}, sectors related by an  $SL(2,\mathbb{Z})$ action are conjectured to be equivalent. This applies in particular to the $SL(2,\mathbb{Z})$ orbit of the $W$-boson sector $(0,1)$, which includes all sectors $(m,n)$ with co-prime integers $m$ and $n$. 

For various reasons, S-duality can, at best, hold in supersymmetric versions of YMH quantum field  theory \cite{Osborn}. The evidence to support the conjecture has therefore mostly come from the consideration of BPS quantum states, whose energy can be computed exactly even with semiclassical or perturbative methods since higher order corrections vanish on account of the supersymmetry~\cite{Sen}.

Returning to the linear spectral problems addressed in the current paper, we note that  both the linear fluctuations around a monopole and eigenmodes of the Laplace operator on the moduli space are, in fact, semi-classical approximations of dyonic states in quantum YMH theory. In particular, after quantisation of an angular collective coordinate which we review in Sect.~\ref{dyons}, fluctuations around a single monopole ($m=1$) may describe states with arbitrary electric charge $n\in \ZZ $. Eigenstates of the Laplace operator on the moduli space of two monopoles ($m=2$) may have arbitrary electric charge $n\in \mathbb{Z}$. Since $(1,n)$, $n\in \ZZ$ and $(2,n)$, $n$ odd, all lie on the $SL(2,\mathbb{Z})$-orbit consisting of co-prime pairs $(m,n)$, the corresponding sectors of YMH theory are related by S-duality. 
  
Since we are not working in an explicitly supersymmetric setting and are looking at bound states which are not of the BPS-type, we do not expect the spectra in these sectors to be related in any simple way. Nonetheless, we find striking qualitative similarities. At the end of this paper, we will discuss them in the light of duality conjectures. 

The paper is organised as follows. In Sect.~2 we establish general results on a self-adjoint extension and the number of bound states of a class of second order differential operators on the half-line $0\leq \rho <\infty$. The class of Schr\"odinger operators we consider contains  a Calogero ($1/\rho^2$) potential  near $\rho=0$ and a Coulombic  ($1/\rho$) tail for $\rho\rightarrow \infty$. The technical assumptions we make are designed to cover the radial  operators obtained from the linearised YMH equations and the Laplace operator on the Atiyah-Hitchin manifold after separation of variables, and require a generalisation of results available in the literature. In Sect.~3 we apply our results to a channel of the linearised YMH equations previously studied in \cite{RS} and \cite{BT}, and prove the existence of infinitely many bound states. By means of suitable trial functions, we give numerical upper bounds for the lowest 20 eigenvalues. In Sect.~4 we turn to the Laplace operator on the Atiyah-Hitchin manifold. Following \cite{M} and \cite{S} to separate radial and angular variables, we focus on three single channels of the Laplace operator. We prove the existence of infinitely many bound states in each of them, compute the eigenvalues numerically and compare with previous numerical results in the literature. Sect.~5 contains a discussion of our results and our conclusions. 

%%%%%%%%%%%%%%%%%%%%%%%%%%%%%%%%%%%%%%%%%%%%%%%%%%%%%%%%%%%%%%

\section{Coulombic bound states on the half-line} \label{generalresults}

Our strategy for showing that the linearised YMH equations around the BPS monopole and the Laplace operator on the Atiyah-Hitchin manifold have infinitely many eigenvalues is to find an orthonormal family of states with energy below the bottom of the essential spectrum. In both cases the associated Hamiltonians reduce to one-dimensional Schr{\"o}dinger operators on the half-line whose potentials have, as leading terms, a combination of a Calogero and a Coulombic potential.  

Schr{\"o}dinger operators on the half-line have been studied extensively in the literature since they arise as the radial part of  Schr{\"o}dinger operators in two- or three-dimensional Euclidean space, see  for example \cite[Appendix to X.I]{RS2}. Denoting the radial coordinate by $\rho$  (the more conventional $r$ is reserved for a different radial coordinate below), the identity
\begin{equation}
\label{radlap}
-\frac{1}{\nu^2} \partial_{\rho} \nu^2 \partial_{\rho} = \frac{1}{\nu }\left(-\partial_{\rho}^2 + \frac{\nu''}{\nu}\right)\nu, 
\end{equation}
for an arbitrary non-vanishing and differentiable function $\nu$ of ${\rho}$, implies that, in any dimension, one can bring the radial derivative appearing in the Laplace operator on Euclidean space into the `flat' form $\partial^2_{\rho}$ at the expense of introducing the effective potential $\nu''/\nu$. 
In the most familiar three-dimensional case, $\nu({\rho})={\rho}$, and the effective potential vanishes. In two dimensions, however, $\nu({\rho})=\sqrt{{\rho}}$, and so the effective potential is the attractive Calogero potential $-1/4{\rho}^2$. 

The unusual feature of the potentials we encounter in this paper is that they combine small ${\rho}$ behaviour which is typical of two-dimensional problems (involving attractive Calogero potentials) with large ${\rho}$ behaviour which is characteristic of three dimensions (for example a $1/{\rho}$ Coulombic potential and a repulsive or `centrifugal' Calogero potential). In this section we establish the framework for studying   the selfadjointness and the spectrum of Hamiltonians of this type.

We write 
\begin{equation}
\left\langle f,g \right\rangle=\int_{0}^{\infty}f({\rho})\overline{g({\rho})}\mathrm{d}{\rho}
\end{equation}
for  the inner product of the space $L^2(0,\infty)$ and  
\begin{equation}
\left\|f\right\|=\sqrt{\left\langle f,f\right\rangle}
\end{equation}
for the corresponding norm. We are interested in potentials on the open half-line which are  real continuous and  have the following  behaviour near  $0$ and $\infty$: 
\begin{equation}   \label{gencondV}
V({\rho})=\begin{cases} \frac{c_2}{{\rho}^{2}}+\operatorname{O}(1), & {\rho}\to 0, \\
C_0+\frac{C_1}{{\rho}}+\operatorname{o}\left(\frac{1}{\rho}\right), & {\rho}\to \infty, \end{cases}
\end{equation}
for  a constant $c_2\in \mathbb{R}$ characterising the asymptotics for small ${\rho}$ and    constants $C_0,C_1\in \mathbb{R}$ characterising the asymptotics for  large ${\rho}$.

Let the differential expression corresponding to the Hamiltonian be given by
\begin{equation}
\label{hamdef}
\tilde{H}=-\frac{d^{2}}{d{\rho}^{2}}+V(\rho),
\end{equation}
with domain $C^{\infty}_{0}(0,\infty)$.
Then $\tilde{H}$ is a densely defined symmetric operator acting on $L^2(0,\infty)$. Below we will fix a specific selfadjoint extension of $\tilde{H}$. We begin by determining conditions on $V$ which are sufficient to ensure that $\tilde{H}$ is semi-bounded.

\begin{lemma} 
Suppose  the potential $V$ in \eqref{hamdef}  can be written as 
$V({\rho})=\frac{c_2}{{\rho}^2}+W({\rho})$ where  $c_2\geq -\frac14$ and the function  $W:[0,\infty)\rightarrow \mathbb{R}$  is bounded.   Then the symmetric operator $\tilde{H}$ is semi-bounded below.
\label{lemma1}
\end{lemma}
\begin{proof}
Let
\bee
      \cmin=\inf_{0\leq {\rho}<\infty} W({\rho}) >-\infty.
\eee
so that 
\bee
      V({\rho})\geq \frac{c_2}{{\rho}^2}+\cmin \qquad \forall {\rho}>0.
\eee
Consider $u \in C^\infty_0(0,\infty)$.  By virtue of Hardy's inequality \cite[Lemma~5.3.1]{EBD}, we deduce
\begin{equation}
\langle (\tilde{H}-\cmin)u,u\rangle \geq \int_0^\infty
\left(|u'({\rho})|^2-\frac{|u({\rho})|^2}{4{\rho}^2}\right)d{\rho} +  \int_0^\infty
\frac{|u({\rho})|^2}{{\rho}^2}\left(\frac14+c_2\right)d{\rho}\geq \int_0^\infty
\frac{|u({\rho})|^2}{{\rho}^2}\left(\frac14+c_2\right)d{\rho}.
\end{equation} 
The condition $c_2\geq -\frac14$ ensures that the right hand side is non-negative. 
\end{proof}

\begin{remark}
If the potential $V$ in \eqref{hamdef} can be written in the slightly more general form 
\begin{equation*}
      V({\rho})=\frac{c_2}{{\rho}^2}+\frac{c_1}{{\rho}}+W({\rho}), \qquad c_1,c_2\in \mathbb{R}, 
\end{equation*}
with $W$ as in Lemma 1, 
then the corresponding operator $\tilde{H}$ is also semi-bounded below 
if either  $c_2>-1/4$ and $-\infty<c_1<\infty$ or $c_2=-1/ 4$ and $c_1\geq0$.
As we will only be concerned with the case $c_1=0$ in our applications below, we omit the proof of this statement, which is an elementary extension of the proof given for Lemma 1. 
\end{remark}

Let us now turn to the question of selfadjoint extensions of $\tilde H$. We follow \cite{BDG}, \cite[\S7.2.3]{GTV}, \cite{2010Duclosetal} and \cite{2010Oliveira}.  We will identify and fix such a selfadjoint extension which, depending on the parameters occurring in $V$ particularly on  $c_2$, may or may not be unique. Throughout we suppose that the potential satisfies the following.

\begin{assumption} \label{condpot}
  The potential  $V:(0,\infty)\rightarrow \mathbb{R}$ can be written in terms of a continuous and bounded  function  $W:[0,\infty)\rightarrow \mathbb{R}$ as 
\begin{equation}
\label{VW}
      V({\rho})=\frac{c_2}{{\rho}^2}+W({\rho}), 
\end{equation}
for some constant $c_2\geq-\frac14$, where 
  \begin{enumerate}
  \item \label{cond1} the limit $\lim_{{\rho}\rightarrow\infty}W({\rho})=:C_0$ exists and
  \item \label{cond2}   $\int_1^\infty (W({\rho})-C_0)^2 d{\rho}<\infty$.
  \end{enumerate}  
\end{assumption}

We now set 
\begin{equation} \label{defc1}
c_2 = m^2-\frac 1 4, \quad   m \geq 0,
\end{equation}
 and denote the Hamiltonian which comprises the leading term of $\tilde{H}$ for small ${\rho}$ by 
\begin{equation}
\tilde{H}_0 = -\frac{d^{2}}{d{\rho}^{2}}+\frac{m^2-\frac 1 4 }{{\rho}^2}, \quad m \geq 0,
\end{equation}
in the same domain $C^\infty_0(0,\infty)$. Then, $\tilde{H}_0$ is essentially selfadjoint for $m\geq 1$ and has deficiency indices $(1,1)$ for  $ 0\leq  m < 1$. We briefly review the reason for this and explain why one of the extensions for $ 0\leq m < 1$ is natural in the present context. 

As already mentioned after \eqref{radlap}, the differential operator $\tilde{H}_0$ arises  as the radial part of the Laplacian on the two-dimensional Euclidean space $\mathbb{R}^{2}$. For Cartesian coordinates $(x_1,x_2)$, and polar coordinates ${\rho}=\sqrt{x_1^2+x_2^2}\in (0,\infty)$ and $\varphi\in [0,2\pi)$,  we explicitly have 
\begin{equation}
\Delta=\frac{ \partial_1^2}{\partial x_1^2}+ \frac{\partial^2}{\partial x_2^2} = \frac{1}{{\rho}} \partial _{{\rho}} {\rho} \partial_{\rho}+\frac{1}{{\rho}^{2}}\partial_{\varphi}^{2}.
\end{equation}  
Therefore solving  $ 
-\Delta\psi =0$               
with the ansatz
\begin{equation}
\label{radanz}
\psi= u({\rho}) e^{im\varphi}
\end{equation}
leads to the radial equation
\begin{equation}
-\frac{1}{{\rho}}\left({\rho}u'\right)' +\frac{m^2}{{\rho}^2} u =0.
\end{equation}
Here $m$ has to be an integer for \eqref{radanz} to be a single-valued function, but this does not affect the essence of the following argument and we will continue to assume that $m$ is a non-negative real number.
The identity \eqref{radlap} shows that the new radial function
\begin{equation}
\label{etau}
\eta({\rho}) =\sqrt{{\rho}} u({\rho})
\end{equation}
then satisfies
\begin{equation}
\label{H0eq}
\tilde{H_0}\eta = 0.  
\end{equation}
 A basis of solutions for this equation is given by
\begin{align}
\label{fundamental}
\eta_1({\rho}) = {\rho}^\frac  12 ,\; \eta_2({\rho})= {\rho}^\frac 12 \ln ({\rho})  \qquad & \text{if} \quad m=0, \nonumber \\
\eta_1({\rho}) = {\rho}^{\frac 12 + m} , \; \eta_2({\rho})= {\rho}^{\frac 1 2 -m} \quad & \text{if} \quad m>0.
\end{align}
When $m \geq 1$ only one of these solutions is square integrable with respect to $d{\rho}$ near ${\rho}=0$, and so ${\rho}=0$ is a limit point for $\tilde{H}_0$ in that case.  

When $0\leq m <1$, however, both $\eta_1$ and $\eta_2$ are square integrable with respect to $d{\rho}$ near ${\rho}=0$,  so that ${\rho}=0$ is a limit circle for the differential expression $\tilde{H}_0$, and an additional boundary condition is required. Both functions $u$ obtained from the fundamental solutions \eqref{fundamental} via \eqref{etau} in this case are square integrable with respect to the radial measure ${\rho}d{\rho}$ on $\mathbb{R}^2$, but  $u_2({\rho}) =\eta_2({\rho})/\sqrt{{\rho}}$  is singular at ${\rho}=0$ and does 
not lie in the domain of the Laplacian. Thus, the requirement that solutions are in the domain of $\Delta$ naturally provides the boundary condition that we set below.

Due to the asymptotic behaviour of the potential, $\rho=\infty$ is a limit point for all values of $m$. We thus fix a selfadjoint extension of $\tilde{H}_0$ as follows. Let $\zeta\in C^\infty([0,\infty))$ be such that $\zeta(\rho)=1$ for $\rho \leq 1$ and  $\zeta(\rho)=0$ for $\rho\geq 2$ and set $\zeta_m(\rho)=\zeta(\rho)\rho^{\frac12 + m}$. Let $\overline{\tilde{H}_0}$ be the closure of the operator $\tilde{H}_0$, the minimal operator associated to $\tilde{H}_0$, and let
\bee
\label{domdef}
   \mathcal{D}=\mathrm{D}\left(\overline{\tilde{H}_0}\right)+ \mathbb{C}\zeta_m.
\eee
We define $H_0$ to be the extension of $\tilde{H}_0$ with domain
$\mathrm{D}(H_0)=\mathcal{D}$. By virtue of \cite[Prop. 4.17]{BDG}, $H_0$ is selfadjoint. Moreover $\mathrm{D}(H_0)=\mathrm{D}\left(\overline{\tilde{H}_0}\right)$ if and only if $m\geq 1$, and in this case $\tilde{H}_0$ is essentially selfadjoint.

\begin{lemma}
Suppose the potential $V$ satisfies the conditions of the Assumption~\ref{condpot}. Then the potential $W({\rho})-C_0$ is relatively compact with respect to $H_0$.
\label{lemma2}
\end{lemma}
\begin{proof}
Let 
\bee
     J_m({\rho})=\sum_{j=0}^\infty \frac{(-1)^j ({\rho}/2)^{2j+m}}{j! \Gamma (j+m
+1)}
\eee
be the Bessel function of the first kind. Let 
\bee
    I_m({\rho})=i^{-m}J_m(i{\rho}), \qquad K_m({\rho})=\begin{cases} \frac{\pi}{2} \frac{I_{-m}({\rho})-I_{m}({\rho})}{\sin(m
\pi)}, &  m\not=0,1,\ldots\\ 
\frac{(-1)^{m-1}}{2}\left(\left.\frac{\partial I_{\nu}(\rho)}{\partial \nu}\right|_{\nu=m} + \left.\frac{\partial I_{\nu}(\rho)}{\partial \nu}\right|_{\nu=-m}   \right), & m=0,1,\ldots
\end{cases}
\eee
be the modified Bessel functions. Then the resolvent $(H_0+1)^{-1}$ is given by the Green's function \cite[\S4.2]{BDG} 
\bee  
       G_m({\rho},{\sigma})=\sqrt{{\rho}{\sigma}}I_m(\min\{{\rho},{\sigma}\}) K_m
(\max\{{\rho},{\sigma}\}).
\eee
Decompose
\bee
      G_m({\rho},{\sigma})=\sum_{k=1}^4 G^k_m({\rho},{\sigma}),
\eee
where
\begin{align*}
    G^1_m({\rho},{\sigma})&=G_m({\rho},{\sigma}) \mathbbm{1
}_{(0,1]}({\rho})\mathbbm{1}_{(0,1]}({\sigma}), \\
    G^2_m({\rho},{\sigma})&=G_m({\rho},{\sigma}) \mathbbm{1}_{(0,1]}({\rho})\mathbbm{1}_{(1,\infty)}({\sigma}), \\
    G^3_m({\rho},{\sigma})&=G_m({\rho},{\sigma}) \mathbbm{1}_{(1,\infty)}({\rho})\mathbbm{1}_{(0,1]}({\sigma}), \\
    G^4_m({\rho},{\sigma})&=G_m({\rho},{\sigma}) \mathbbm{1}_{(1,\infty)}({\rho})\mathbbm{1}_{(1,\infty)}({\sigma}).
 \end{align*}
Then \cite[(4.10)]{BDG}
\begin{align*}
  |G_m^1({\rho},{\sigma})|&\leq \begin{cases} a_{10}(\min\{{\rho},{\sigma}\})^{\frac12}|\ln (\max\{{\rho},{\sigma}\})|, & m=0, \\
a_{1m} (\min\{{\rho},{\sigma}\})^{\frac12+m}(\max\{{\rho},{\sigma}\})^{\frac12-m}, & m>0,
\end{cases} \\
|G_m^2({\rho},{\sigma})|& \leq a_{2m} {\rho}^{\frac12+m}e^{-{\sigma}}\mathbbm{1}_{(0,1]}({\rho})\mathbbm{1}_{(1,\infty)}({\sigma}), \\
|G_m^3({\rho},{\sigma})|& \leq a_{3m} e^{-{\rho}}{\sigma}^{\frac12+m}\mathbbm{1}_{(1,\infty)}({\rho})\mathbbm{1}_{(0,1]}({\sigma}), \\
|G_m^4({\rho},{\sigma})|& \leq a_{4m} e^{-|{\rho}-{\sigma}|}\mathbbm{1}_{(1,\infty)}({\rho})\mathbbm{1}_{(1,\infty)}({\sigma}), 
\end{align*}
where $a_{im}>0$ are constants depending on $m$. 

Let
\bee
     \mathcal{J}_k=\int_0^\infty \int_0^\infty |(W({\rho})-C_0) G_m^k({\rho},{\sigma})|^2 d{\rho} d{\sigma},
\eee
so that
\bee
   \int_0^\infty \int_0^\infty |(W({\rho})-C_0) G_m({\rho},{\sigma})|^2 d{\rho} d{\sigma}=\sum_{k=1}^4 \mathcal{J}_k.
\eee
Then
\begin{align}
\mathcal{J}_{2}\leq a^{2}_{2m}\int^{1}_{0}|W({\rho})-C_0|^{2}{\rho}^{1+2m}d{\rho}\int^{\infty}_{1}e^{-2{\sigma}}d{\sigma}<\infty,
\end{align}
since $W$ is continuous on $[0,\infty)$. 
The square integrability of $W-C_0$ on $[1,\infty)$ ensures that 
\begin{align}
\mathcal{J}_{3}\leq a^{2}_{3m}\int^{1}_{0}{\sigma}^{1+2m}d{\sigma}\int^{\infty}_{1}|W({\rho})-C_0|^{2}e^{-2{\rho}}d{\rho}  <\infty,
\end{align}
and
\begin{align}
\mathcal{J}_{4}&\leq a^{2}_{4m}\int^{\infty}_{1}\int^{\infty}_{1}|W({\rho})-C_0|^{2}e^{-2|{\rho}-{\sigma}|}d{\sigma}d{\rho}\nonumber\\
&={a}^{2}_{4m}\int^{\infty}_{1}|W({\rho})-C_0|^{2}\left(e^{-2{\rho}}\int^{{\rho}}_{1}e^{2{\sigma}}d{\sigma}+e^{2{\rho}}\int^{\infty}_{{\rho}}e^{-2{\sigma}}d{\sigma}\right)d{\rho}\nonumber\\
&\leq {a}^{2}_{4m}\int^{\infty}_{1}|W({\rho})-C_0|^{2}d{\rho}<\infty.
\end{align}
Finally, using  $|W({\rho})-C_0|^2 < B $ for  ${\rho} \in [0,1]$ and some positive real number $B$, we consider the term 
for $k=1$. If   $m>0$ we have
\begin{align}
\mathcal{J}_{1}&\leq a^{2}_{1m}\int^{1}_{0}|W({\rho})-C_0|^{2}\left({\rho}^{1-2m}\int^{{\rho}}_{0}{\sigma}^{1+2m}d{\sigma}+{\rho}^{1+2m}\int^{1}_{{\rho}}{\sigma}^{1-2m}d{\sigma}\right)d{\rho}\nonumber\\
& \leq a^{2}_{1m}B \int^{1}_{0} \left( {\rho}^{1-2m}\int^{{\rho}}_{0}{\sigma}^{1+2m}d{\sigma}+{\rho}^{1+2m}\int^{1}_{{\rho}}{\sigma}^{1-2m}d{\sigma}\right) d{\rho}\  \nonumber \\
&=  \frac{ a^{2}_{1m}B}{4(m+1)}. 
\end{align} 
The computation of the integral in the last step is elementary, but has to be carried out separately for the case $m=1$. However, the answer agrees with the general formula given above.
If  $m=0$ we have
\begin{align}
\mathcal{J}_{1}&\leq a^{2}_{10}\int^{1}_{0}|W({\rho})-C_0|^{2}\left(\ln ^2 ({\rho})\int^{{\rho}}_{0}{\sigma}d{\sigma}+{\rho}\int^{1}_{{\rho}}\ln^{2} ({\sigma})d{\sigma}\right)d{\rho}\nonumber\\
&\leq {a}^{2}_{10}B \int^{1}_{0}\left( -\frac 12  {\rho}^2 \ln^2 ({\rho})+ 2 {\rho}^2 \ln({\rho}) + 2 {\rho} -2 {\rho}^2\right)  d{\rho}\nonumber\\
&=\frac{ 2 {a}^{2}_{10}B}{27},
\end{align}
where we have again omitted steps in an elementary integration. 

Summing up, for $m\ge0$ we have $\mathcal{J}_k<\infty$ under the assumptions of the lemma. Therefore, the operator 
 \bee
(W-C_0)(H_0+1)^{-1}
\eee 
is Hilbert-Schmidt and hence compact. This ensures that indeed $W-C_0$ is relatively compact with respect to $H_0$. 
\end{proof}

\begin{definition}
\label{extdef}
With $\tilde{H}$ of the form \eqref{hamdef} and $V$ satisfying the Assumption~\ref{condpot}, we denote by $H$ the extension of the operator $\tilde{H}$ to the domain $\mathrm{D}(H)=\mathcal{D}$ defined in \eqref{domdef}. 
\end{definition}

By virtue of \cite[Corollary~2 to Theorem~XIII.14]{RS4}, $H$ is always selfadjoint. Moreover, $H$ is semi-bounded below and
\bee
      \sigma_{\mathrm{ess}} (H)=[C_0,\infty).
\eee
We now establish conditions for $H$ to have infinitely many bound states.

 \begin{theorem} If  Assumption~\ref{condpot} is satisfied, and moreover $V$ has the asymptotic expansion 
\begin{equation}
V({\rho}) = C_0 + \frac{C_1}{{\rho}} + o\left( \frac{1}{\rho}\right), \quad \text{for} \quad {\rho}\rightarrow \infty,
\end{equation} 
with $C_1<0$,  then $H$ has infinitely many eigenvalues below $C_0$.
\label{lemma3}
\end{theorem}
\begin{proof}
 We use a similar argument to the one considered in \cite[\S 8.3]{GS}.
By assumption, we can write the potential as 
\bee
V(\rho) = C_0 + \frac{C_1 +f(\rho)}{\rho},
\eee
in terms of a continuous function $f$ on $(0,\infty)$ satisfying 
\bee
 \lim_{\rho\rightarrow \infty}  f(\rho) =0.
 \eee
We pick a $u\in C^\infty_{0}(0,\infty)$, such that $\operatorname{supp}(u)\subset(1,2)\hspace{1.5mm}\text{and}\hspace{1.5mm}\left\|u\right\|=1$, and set 
\bee
u_{n}(\rho)=2^{-n/2}u(2^{-n}\rho),\hspace{5mm}n=0,1,...
\eee
Then the $u_n$ have non-overlapping support and satisfy the orthonormality condition
\bee
 \left\langle u_{n},u_{m}\right\rangle=\delta_{mn}.
\eee
Changing variables to ${\sigma}=2^{-n}\rho$, we calculate 
\begin{align}
\label{ineq}
\langle (H-C_0) u_n,u_n\rangle &= \int_{2^n}^{2^{n+1}} \left( \left|\frac{d u_n}{d\rho}(\rho) \right|^2+ \frac{C_1 +f(\rho)} {\rho} |u_n(\rho)|^2 \right)d \rho \nonumber \\
& = 2^{-2n}\int_{1}^{2} |u'({\sigma})|^2 d \sigma +
2^{-n}\int_{1}^{2} \frac{C_1+f(2^n\sigma) } {\sigma}|u(\sigma)|^2 d\sigma.
\end{align}
Now we can pick $n_0$  so that, for all $\sigma \in [1,2]$ and  $n>n_0$, we have  $|f(2^n\sigma)|<C_1$. If $C_1<0$ it follows that 
\bee
\int_{1}^{2} \frac{C_1+f(2^n\sigma) } {\sigma}|u(\sigma)|^2 d\sigma <0 \quad \text{for} \; n > n_0.
\eee
Thus the last line of \eqref{ineq} is a sum of a positive and a negative term when $C_1<0, n> n_0$. Since the positive term decreases with $n$ faster than the negative term, we can make the sum negative by choosing $n$ big enough. Then $ \langle (H-C_0) u_n,u_n\rangle < 0$ and therefore $\langle Hu_n,u_n\rangle$ lies below the essential spectrum for sufficiently large $n$. This is enough to ensure that we have an infinite number of negative eigenvalues as a consequence of the Rayleigh-Ritz principle, cf. \cite[\S8.3]{GS}. 
\end{proof}

%%%%%%%%%%%%%%%%%%%%%%%%%%%%%%%%%%%%%%%%%%%%%%%%%%%%%%%%%%%%%%%

\section{Bound states in the linearised YMH equations}
\subsection{BPS monopoles and dyons}
\label{dyons}
The first example of a Schr\"odinger operator on the half-line to which we apply the results of the previous section arises in YMH theory in 3+1 dimensional Minkowski space. Our main reference for the derivation of this Schr\"odinger operator is the paper 
 \cite{RS}, to which we refer for details and background.  

The YMH model consists of a non-abelian gauge potential $A= A_0 dt + A_1dx_1 +A_2dx_2 +A_3 dx_3$ coupled to a Higgs field $\phi$, both taking values in the Lie algebra $su(2)$. Writing $\textbf{x}= (x_1,x_2,x_3)\in \mathbb{R}^3$, and $\partial_i =\partial/\partial x_i$, $i=1,2,3$, we only require the spatial covariant derivatives $D_i=\partial_i+e \left[A_i,\cdot \right]$ and the spatial components of the Yang-Mills field strength tensor $F_{ij}=\partial_i A_j-\partial_jA_i+e\left[A_i,A_j\right]$, where $[ \cdot, \cdot ]$ is the Lie algebra commutator, $e$ is the Yang-Mills coupling constant and $i,j=1,2,3$. 

In the following we set the coupling constant $e$ to one and consider a particular limit of the theory, called the BPS limit, where the self-coupling of the Higgs field vanishes. In that limit, the second order static YMH equations are implied by the first order BPS equations 
\bee
\label{BPSeq}
D_i\phi= \frac  12 \epsilon_{ijk} F_{jk},
\eee  
where $i,j,k=1,2,3$ and repeated indices are summed over. 

In terms of the a basis $t_1,t_2,t_3$ of the $su(2)$ Lie algebra satisfying $[t_i,t_j]=\epsilon_{ijk} t_k$, a particular solution of the BPS equations is the spherically symmetric BPS monopole: 
 \begin{equation}
A_{i}(\textbf{x})=\left(\frac{1}{r^{2}}-\frac{1}{r\sinh(r)}\right)(\textbf{x}\times \textbf{t})_i,\hspace{3mm}\phi(\textbf{x})=\frac{1-r\coth(r)}{r^{2}}\textbf{x}\cdot\textbf{t}, 
\label{eq:26}
\end{equation}
where $ r =\sqrt{x_1^2+x_2^2+x_3^2}$ and we have denoted Euclidean vectors by bold letters. 
The invariance of the BPS equations under Euclidean symmetries implies that translating this solution in space gives an $\RR^3$ worth of solutions. Gauge invariance of the equations  means that $SU(2)$ gauge transformations map solutions to solutions. In gauge theories, solutions related by gauge transformations which are the identity at infinity are generally considered equivalent, but gauge transformations which are non-trivial at infinity may have a physical significance. Such gauge transformations are often called `large', and will play a role in the interpretation of our results. We therefore  review them briefly. 

As discussed and explained in \cite{GM}, for monopoles in $SU(2)$ YMH theory, there is essentially a circle of large gauge transformations, generated by the Higgs field itself and parameterised by an angle $\chi \in [0,2\pi)$:
\bee
\label{largeg}
g_\chi (\textbf{x}) = \exp (\chi \phi(\textbf{x}) ).
\eee 
Note that, for large $r$, $\phi(\textbf{x}) \approx -\hat{\textbf{x}}$, so that $g_\chi (\textbf{x})$ is asymptotically a rotation about $\textbf{x}$. 
Acting with translations and  the large gauge transformations \eqref{largeg}, and dividing out by small gauge transformations yields the moduli space of monopoles of magnetic charge one:
\bee
\label{mod1}
M_1= \RR^3 \times S^1.
\eee   
The physical significance of the angular variable $\chi$  on $S^1$ becomes manifest when it varies in time. The monopole then acquires an electric charge proportional to the angular speed $\dot{\chi}$, thus turning into a dyon. Quantum states are given by wave functions on the moduli space \cite{GM}. A wave function of the form $\exp(in\chi)$, with $n\in \ZZ $ describes a dyon of electric charge $n$.

 \subsection{Linearising around the  BPS monopole}
 \label{linearised}
The Schr\"odinger operator we would like to study is obtained by linearising the general, time-dependent YMH equations around the static configuration \eqref{eq:26}. The stationary ansatz 
\begin{equation}
A_{i}(t,\mathbf{x})=A_{i}(\mathbf{x})+a_{i}(\mathbf{x})e^{i\omega t},\hspace{5mm}\phi(t,\mathbf{x})=\phi(\mathbf{x})+\varphi(\mathbf{x})e^{i\omega t},
\label{eq:240}
\end{equation}
considered in \cite{RS}, leads to the  following  coupled partial differential equations on Euclidean $\mathbb{R}^3$:
\begin{align}
\label{eq:241}
D_{i}D_{i}\varphi+\left[a_{i},D_{i}\phi\right]+D_{i}\left[a_{i},\phi\right]&=-\omega^{2}\varphi, \nonumber \\
D_{i}D_{i}a_{j}-D_{i}D_{j}a_{i}+\left[a_{i},F_{ij}\right]&=\left[\phi,D_{j}\varphi\right]+\left[\varphi,D_{j}\phi\right]+\left[\phi,\left[a_{j},\phi\right]\right]-\omega^{2}a_{j}.
\end{align}

Exploiting the (suitably defined) rotational symmetry of  the BPS monopole and focusing on the vanishing total angular momentum, the ansatz
\begin{align}
\label{eq:ansatz2}
\varphi(\mathbf{x})=0, \quad a_i(\mathbf{x})=\frac{1}{r}\left(v(r) ((\hat{\mathbf{x}}\cdot\mathbf{t})\hat{x}_i - t_i)  +\sqrt{2}\alpha(r) (\hat{\mathbf{x}}\cdot\mathbf{t})\hat{x}_i\right),
\end{align}
involving two unknown functions of the radial coordinate $r$, leads to the following set of ordinary differential equations  
\begin{align}
\label{valpha}
\left(-\frac{d^{2}}{dr^{2}}+\frac{3}{\sinh^{2}(r)}-\frac{2\coth(r)}{r}+\coth^{2}(r)\right)v+\frac{2\sqrt{2}\coth(r)}{\sinh(r)}\alpha &=\omega^{2}v, \nonumber \\
\left(-\frac{d^{2}}{dr^{2}}+\frac{2}{\sinh^{2}(r)}+\frac{2}{r^{2}}\right)\alpha+\frac{2\sqrt{2}\coth(r)}{\sinh(r)}v&=\omega^{2}\alpha.
\end{align}
As explained in \cite{RS}, this system of equations can be decoupled when $\omega\neq 0$, and brought into the form 
\begin{align}
\label{xizeta}
\left(-\frac{d^{2}}{dr^{2}}+\frac{1}{\sinh^{2}(r)}+\frac{2}{r^{2}}-\frac{2\coth(r)}{r}+\coth^{2}(r)\right)\xi &=\omega^{2}\xi,\nonumber \\
\left(-\frac{d^{2}}{dr^{2}}+\frac{2}{\sinh^{2}(r)}\right)\zeta&=\omega^{2}\zeta,
\end{align}
via the (invertible) transformation
\begin{align}
\label{changevar}
\omega\alpha &=\frac{d\zeta}{dr}-\frac{\zeta}{r}+\frac{\sqrt{2}}{\sinh(r)}\xi,\nonumber \\
\omega v&=-\frac{d\xi}{dr}-\frac{1-r\coth(r)}{r}\xi-\frac{\sqrt{2}}{\sinh(r)}\zeta,
\end{align}
provided $\omega \neq 0$. 
We note that the equations \eqref{valpha} also  have a zero-energy solution which can easily be given explicitly, see e.g. \cite{RS}. Since $\omega=0$ for this solution, it cannot be obtained from solutions of the system \eqref{xizeta} which we study here.

As also explained in \cite{RS}, the equation for $\zeta$ does not have bound states, and moreover gauge invariance requires that $\zeta =0$. 
Therefore we focus on the equation for $\xi$, which is a Sturm-Liouville problem for the  Schr\"odinger operator 
\begin{equation}
\tHYMH=-\frac{d^{2}}{dr^{2}}+\VYMH
\label{pot3_24}
\end{equation}
with the potential
\begin{equation}
\VYMH(r)=\frac{1}{\sinh^{2}(r)}+\frac{2}{r^{2}}-\frac{2\coth(r)}{r}+\coth^{2}(r).
%&=\frac{2}{r^{2}}-\frac{2}{r}+\frac{1+\cosh^{2}r}{\sinh^{2}r}+\frac{2}{r}(1-\coth r)
\label{pot1}
\end{equation}
This potential has the asymptotic expansion
\begin{equation}   \label{V2exp}
\VYMH(r)=\begin{cases} \frac{2}{r^{2}}+\operatorname{O}(1), & r\to 0, \\
1-\frac{2}{r}+ \frac{2}{r^{2}}+ \operatorname{O}\left(\frac{1}{r^3}\right), & r\to \infty, \end{cases}
\end{equation}
so, in particular, is of the form \eqref{gencondV} with $c_2=2$, $C_0=1$ and $C_1=-2$. It is easy to check that it satisfies the assumptions of Theorem~\ref{lemma3} (and {\em a fortiori} those of  Lemmas \ref{lemma1} and \ref{lemma2}). As $m=3/2>1$ in \eqref{defc1}, we have the following result. 

\begin{corollary} \label{coro3}
The closure $\HYMH$ of $\tHYMH$ acting on $L^2(0,\infty)$ is selfadjoint, its essential spectrum is the segment $[1,\infty)$ and it has infinitely many eigenvalues below $1$. 
\end{corollary}

The infinitely many eigenstates of the Schr\"odinger operator \eqref{pot3_24} define radially symmetric fluctuations around the BPS monopole via substitution into the equations \eqref{changevar} and \eqref{eq:ansatz2}. They therefore belong to the sector with magnetic change $m=1$. However, as explained in \cite{RS}, the fluctuations do not have well defined electric charge.

To obtain eigenstates of the electric charge operators, one needs to include the collective angular coordinate $\chi$ for large gauge transformations \eqref{largeg} in the discussion.
These gauge transformations act on the underlying BPS monopoles \eqref{eq:26} but also, by conjugation, on the fluctuations \eqref{eq:ansatz2}, rotating them to $\varphi(\mathbf{x})_\chi=0$ and 
\bee
 a_i(\mathbf{x})_\chi=\left(\frac{v(r)+\sqrt{2}\alpha(r)}{r}\right)(\hat{\mathbf{x}}\cdot\mathbf{t})\hat{x}_i + \cos(\chi|\phi|(\hat{\mathbf{x}}))\left( \frac{v(r)}{r}\right)
  ((\hat{\mathbf{x}}\cdot\mathbf{t})\hat{x}_i - t_i)  
  +\sin(\chi|\phi(\hat{\mathbf{x}})|) (\mathbf{t} \times \hat{\mathbf{x}})_i. 
\eee
Quantum states of definite electric charge are obtained by taking superpositions of these states, weighted with the dyonic wave function $\exp(in\chi)$. This is essentially a Fourier transform from the angular variable $\chi$ to the integer label $n$, and  entirely analogous to the standard interpretation of the moduli space wave functions, which describe superpositions of BPS monopoles. The electric charge  $n$ is arbitrary, so we replicate the  infinitely many eigenstates of the Schr\"odinger operator \eqref{pot3_24} 
 in each of the dyonic sectors $(1,n)$, $n\in \ZZ$.

\subsection{Upper bounds on the eigenvalues}

Having shown that $\HYMH$ has infinitely many negative eigenvalues we would like to find good numerical approximations to the first few of them, improving on previous numerical work in \cite{BT} and \cite{RS} which relied on shooting methods. For this purpose, we employ the Riesz-Galerkin method with a basis which exploits the fact that the potential $\VYMH$ approaches at infinity the potential
\begin{equation}
\label{Coulombpot}
V_C(r)= \frac{\alpha}{r} + \frac{l(l+1)}{r^2} +1, 
\end{equation}
which combines a Coulomb potential with a centrifugal potential for orbital angular momentum quantum number $l$. 
In order to match the asymptotic form of $\VYMH$ we pick $\alpha =-2$ and $l=1$, so that 
\begin{equation}
\VYMH(r) = V_C(r) + \left(    \frac{1+\cosh^{2}(r)} {\sinh^{2}(r)}    +   \frac{2-2\coth (r)}{r}
      \right).
\end{equation}

 The eigenfunctions of the radial Coulomb Hamiltonian with differential expression
 \begin{equation*}
H_C = -\frac{d^2}{dr^2} +V_C
\end{equation*}
are well-known. They are defined in terms of the associated Laguerre polynomial $L^{N}_{k}(r)$
\begin{equation}
L^{N}_{k}(r)=(-1)^N \frac{d^{N}}{dr^{N}}L_{k+N}(r)  
\end{equation} 
where $L_{k}(r)=e^{r}\frac{d^{k}}{dr^{k}}\left(r^{k}e^{-r}\right)$ is the $k$-th Laguerre polynomial. The eigenfunctions of $H_C$ with $l=1$ and $\alpha=-2$, satisfying
\begin{equation}
\label{Coulombenergies}
H_C\U_{n}= E_n\U_{n}, \qquad \qquad E_n =-\frac{1}{n^2},
\end{equation}
 are 
\begin{equation} \label{confusiona0}
\U_{n}(r)=N_{n}^{-1/2} r^2 e^{-\frac{r}{n}}L^{3}_{n-2}\left(\frac{2r}{n}\right).
\end{equation}
Here $N_{n}$ is a normalisation constant ensuring $\|\U_{n}\|=1$ and the principal quantum number takes the values $n=2,3,...$.

The set $\{\U_{n}\}_{n=2}^\infty$ is an orthonormal basis for $L^{2}(0,\infty)$. In order to implement the Riesz-Galerkin method, we pick a finite-dimensional subspace
\bee
    \operatorname{Span} \{\tilde{b}_{n}\}_{n=1}^{\mathfrak{d}}\subset \operatorname{D}(H)
\eee
of dimension $\mathfrak{d}$, where
\begin{equation} \label{eqforbtilde} 
\tilde{b}_{n}(r)=N_{n+1}^{1/2}\U_{n+1}(r)=\left(\sum^{n-1}_{p=0}d^{n+1}_{p}r^{p+2}\right)e^{-\frac{r}{n+1}}
\end{equation}
for suitable coefficients $d^{n+1}_{p}\in \mathbb{R}$.
We then compute the mass matrix                         
\bee
    M=[M(i,j)]_{i,j=1}^{\mathfrak{d}}=\operatorname{diag}(N_2,\ldots,N_{\mathfrak{d}}),\qquad  
M(i,j)=\int^{\infty}_{0}\tilde{b}_{i}(r)\tilde{b}_{j}(r)dr,
\eee
and the stiffness matrix                                 
\bee \label{stifma}
    S=[S(i,j)]_{i,j=1}^{\mathfrak{d}},\qquad  
S(i,j)=\int^{\infty}_{0}\HYMH \tilde{b}_{i}(r)\tilde{b}_{j}(r)dr.
\eee
According to the Rayleigh-Ritz principle, the $k$th negative eigenvalue of 
$S\underline{u}=\nu M\underline{u}$ is an upper bound for the $k$th negative eigenvalue of $ \HYMH$. Further details on the computation of the entries of $S$ are given in Appendix~\ref{massandstiffness}.

For a basis of dimension $\mathfrak{d}=20$ we obtain the results shown in Table~\ref{tab:lambdaI} for the first eleven eigenvalues. We saw convergence of our computations up to single precision as we increased the size of $\mathfrak{d}$ from $1$ to $20$. Additionally we have the extrapolated results obtained for the first eleven eigenvalues for a basis of dimension $\mathfrak{d}=1000$ via linear interpolation.  In \cite{RS},  approximation to these eigenvalues were found via a shooting method and they appear to be below  those found in Table~\ref{tab:lambdaI}. There is no guarantee that the former are above the true eigenvalues of $\HYMH$, whereas the latter certainly are, due to the Rayleigh-Ritz principle.

\begin{table}
\centering
\setlength\tabcolsep{8pt}
\begin{tabular}{c|c|c|c}
\thead{n} & \thead{$\lambda_{n}$ \\ for basis \\ $\mathfrak{d}=20$} & \thead{Extrapolated \\ value of \\ $\lambda_{n}$ for \\ $\mathfrak{d}=1000$} &  \thead{Coulomb \\ approx-\\imation of \\ $1-\frac{1}{(n+1)^2}$}\\
\hline
 1   &  0.773243   &  0.772215  &  0.750000 \\
 2   &  0.897117  &  0.896315  &  0.888889 \\
 3   &  0.941347  &  0.940714  &  0.937500 \\
 4   &  0.962124  &  0.961609  &  0.960000 \\
 5   &  0.973529  &  0.973099  &  0.972222 \\
 6   &  0.980458  &  0.980094  &  0.979592 \\
 7   &  0.984983  &  0.984669  &  0.984375 \\
 8   &  0.988100  &  0.987825  &  0.987654 \\
 9   &  0.990179  &  0.990096  &  0.990000 \\
 10  &  0.990339  &  0.991761  &  0.991736 \\
 11  &  0.993265  &  0.993174  &  0.993056 \\

\end{tabular}
\caption{The values of $\lambda_{n}$ for $\mathfrak{d}=20$, the extrapolated values for $\mathfrak{d}=1000$ and the approximation based on \eqref{Coulombenergies}.}
\label{tab:lambdaI}
\end{table}

\section{The Laplace operator on the moduli space of two monopoles}
\subsection{The Atiyah-Hitchin metric and its asymptotic forms}

The BPS monopole \eqref{eq:26} and the moduli space \eqref{mod1} of charge one magnetic monopoles have remarkable generalisations for higher magnetic charges.  
Following the discovery of the charge one solution, there was rapid progress in constructing various solutions of higher magnetic charge. It is now well-understood that, for given magnetic charge $k$, there is in fact a $4k$-dimensional family of static monopole solutions which constitute the so-called moduli space of charge $k$ monopoles \cite{AHbook}. 

The basic, physical reason for the existence of so many static solutions is that, in the BPS limit, monopoles do not exert any forces on each other so that they can be `superimposed' with arbitrary values of the individual positions and phases. The interpretation of the $4k$ parameters in the moduli space as giving the positions and phases of  $k$  monopoles works well for well-separated monopoles. However, when the monopoles are close together they deform each other and  become bound states with a rich and complicated geometry. All of this is captured by the moduli spaces.  

The moduli spaces inherit a Riemannian metric from the kinetic energy functional of YMH theory. It was first argued by Manton in \cite{Manton} that geodesic motion on the moduli space, equipped  with this metric, is a good approximation to the dynamics of monopoles, provided they are moving sufficiently slowly. This is essentially an adiabatic approximation, where the time evolution is via a sequence of static equilibrium configurations. It was subsequently shown by Atiyah and Hitchin \cite{AHbook} that the moduli space metric is hyperk\"ahler. Combined with symmetry considerations, this is sufficient to determine the moduli space metric for monopoles of  charge two. In that case, the moduli space is eight dimensional, and has the form
\begin{equation}
\label{modmop}
M_2 = \mathbb{R}^3 \times \frac{M^{0}_{2}\times S^1}{\ZZ_2}.
\end{equation}
The $\mathbb{R}^3$- and $S^1$- factors describe the centre-of-mass motion of the two monopoles and carry flat metrics. The manifold  $M^{0}_{2}$ describes the interesting, relative motion of the two monopoles, and we refer to it as the Atiyah-Hitchin manifold in the following. However,  the reader should be aware that some authors reserve this name for the quotient $\MAH/\ZZ_2$. The manifold $\MAH$ is simply-connected and homotopic to a 2-sphere. The metric on $\MAH$ is also hyperk\"ahler. In four dimensions, the hyperk\"ahler property is equivalent to self-duality of the Riemann tensor so that the Atiyah-Hitchin manifold is an example of a gravitational instanton.   

The group $SO(3)$ of spatial rotations  is a symmetry group of YMH theory and acts isometrically on the Atiyah-Hitchin manifold. Therefore, it is convenient to parametrise the Atiyah-Hitchin manifold in terms of this $SO(3)$ action and one transverse radial or `shape' coordinate. The latter parametrises a 1-parameter family of two-monopole configurations which includes two well-separated monopoles where the shape parameter is simply the distance between the two mono\-poles. However, when the two monopoles get close, they deform each other until they coalesce to a doughnut-shaped configuration. The $SO(3)$ orbits are generically isomorphic to $SO(3)/\mathbb{Z}_2$, but the orbit of the doughnut-shaped configuration is exceptional and isomorphic to $S^2$, called the core in the following.

The metric on the Atiyah-Hitchin manifold $M^{0}_{2}$ is most conveniently written in terms of left-invariant $1$-forms $\sigma_1,\sigma_2$ and $\sigma_3$ on $SO(3)$, see \cite{GM} for details. Denoting the transverse coordinate by $r$, the metric takes the Bianchi IX  form 
\begin{equation}
ds^{2}=f^{2}dr^{2}+a^{2}\sigma^{2}_{1}+b^{2}\sigma^{2}_{2}+c^{2}\sigma^{2}_{3},
\end{equation}
where $f,\, a,\, b$ and $c$ are functions of $r$. The self-duality of the metric implies
\begin{equation}
\label{SDcond}
\frac{2bc}{f}\frac{da}{dr}=\left(b-c\right)^{2}-a^{2},
\end{equation}
and two other related equations obtained by cyclic permutation of $a,b,c$. The function $f$ can be chosen to fix the radial coordinate $r$. Following \cite{GM}, we pick
\begin{equation}
f=-\frac{b}{r}.
\end{equation} 
The initial conditions for the coefficient functions are 
\begin{equation}
a(\pi)=0, \quad b(\pi)=\pi, \quad c(\pi)=-\pi.
\end{equation}
The unique solution with these initial conditions can be written in terms of elliptic functions as follows. Let 
\begin{equation}
\label{rbeta}
r=2K\left(\sin \frac{\beta}{2}\right), \qquad  0\leq \beta \leq \pi,
\end{equation}
where  $K$ is the elliptic integral
\begin{equation}
K(k) =\int_0^{\frac{\pi}{2}} \frac{ d\tau }  {\sqrt{ 1-k^2\sin^2 \tau}}.
\end{equation}
With 
\begin{equation}
\label{abcw}
w_1=bc,\quad w_2 =ca, \quad w_3=ab,
\end{equation}
the solution is then given by
\begin{align}
\label{ws}
w_1 (r)&=-\sin\beta \, r\frac{dr}{d\beta}-\frac{1}{2}(1+ \cos\beta) r^2,\nonumber \\
w_2 (r)&=-\sin\beta \, r\frac{dr}{d\beta},\nonumber \\
w_3 (r)&= -\sin\beta \, r\frac{dr}{d\beta}+\frac{1}{2}(1- \cos\beta) r^2.
\end{align}

It turns out that $b>a>0$ away from the core, and that $c$ is negative (in fact,  $c<-2$).  
Defining a proper radial distance coordinate $R$ via 
\begin{equation}
\label{radial}
R(r)=\int_\pi^r-f(\rho) \,d\rho = \int_\pi^r \frac{b(\rho)}{\rho} \,d\rho, 
\end{equation}
 we have the following behaviour near the core   
 \bee
 R=  (r - \pi) +\operatorname{O}\left((r-\pi)^2\right). 
 \eee
This allows us to deduce expansions for the coefficient functions of the Atiyah-Hitchin metric \cite{GM,S}: 
\begin{equation}
\label{smallR}
a(r(R))=  2R  +\operatorname{O}\left( R^2 \right),\quad 
b(r(R))= \pi+\frac{1}{2}R +\operatorname{O}\left( R^2 \right), \quad 
c(r(R))= - \pi+\frac{1}{2}R+\operatorname{O}\left( R^2 \right).
\end{equation}

For large $r$, the coefficient functions $a,b $ and $c$ can be approximated by the functions $\atn,\btn $ and $\ctn$ given by 
\begin{align}
\label{TNapprox}
 \atn(r)=\btn(r)= r\sqrt{1-\frac{2}{r}}, \qquad \ctn(r)= -\frac{2}{\sqrt{1-\frac{2}{r}}}.
\end{align}
Theses functions are exact solutions of the self-duality equations \eqref{SDcond}, and give rise to another hyperk\"ahler  metric, called the Taub-NUT metric with negative `mass' parameter. This metric has $U(2)$ rather than $SO(3)$ symmetry. In the form given above, the metric is  degenerate at $r=2$ and changes signature form $(+,+,+,+)$ to $(-,-,-,-)$ as one crosses from $r>2$ to $r<2$. 
For later use, we note that, as explained in  \cite{GM}, it follows from \eqref{abcw} and \eqref{ws} that  
\begin{align}
\label{larger}
a(r) & = \atn(r) + \operatorname{O}\left( r^2e^{-r} \right), \nonumber \\
b(r) & = \btn(r) + \operatorname{O}\left( r^2e^{-r} \right), \nonumber \\
c(r) & = \ctn(r) + \operatorname{O}\left(e^{-2r} p(r) \right),
\end{align}
where $p$ is an algebraic function of $r$.

In our study of the  spectrum of the Laplace operator on the Atiyah-Hitchin manifold we need the asymptotics of $a,b$ and $c$ both as a function of $r$ and as a function of the proper radial distance $R$. Substituting the asymptotic expressions \eqref{TNapprox} into the definition \eqref{radial} one finds (see also \cite{S}) that      
\bee
\label{Rr}
R(r) = r + \ln r +  \operatorname{O}(1).
\eee

\subsection{The Laplace operator on the Atiyah-Hitchin manifold}

One may approximate the quantum mechanics of $k$ interacting monopoles  at low energy by solving the Schr\"odinger equation on the moduli space  of charge $k$ monopoles, taking the covariant Laplace operator associated to the Riemannian metric as the Hamiltonian. For details of this programme we refer the reader to \cite{GM},  where it is explained and applied to the asymptotic form of the manifold $\MAH$, and to \cite{S} where bound states and scattering states on $\MAH$ are discussed in detail, using a combination of numerical and semiclassical techniques.

The wave function for a two-monopole  quantum state is a $\CC$-valued function on the moduli space \eqref{modmop}. However, assuming without loss of generality that we work in the centre-of-mass frame of the two monopoles we can neglect the dependence on $\RR^3$. Introducing an angular coordinate $\chi\in[0,2\pi)$ on $S^1$, the Hamiltonian is then 
\begin{equation}
H=-\frac{\hbar^2}{16\pi} \frac{\partial^2}{\partial\chi^2} -\frac{\hbar^2}{4\pi}\Delta_{\text{\tiny AH}},
\end{equation}
where $\Delta_{\text{\tiny AH}}$ is the covariant Laplace operator on $\MAH$. It can be written in terms of the 
left-invariant (and right-generated) vector fields $\xi_1,\xi_2$ and $\xi_3$ on $SO(3)$ which are dual to the forms $\sigma_1,\sigma_2$ and $\sigma_3$ (see again \cite{GM} for details). Then
\begin{equation}
\Delta_{\text{\tiny AH}}=
\frac{1}{abcf}\frac{\partial}{\partial r}\left(\frac{abc}{f}\frac{\partial }{\partial r}\right)+\frac{\xi_{1}^{2}}{a^{2}}+\frac{\xi_{2}^{2}}{b^{2}}+\frac{\xi_{3}^{2}}{c^{2}}.
\end{equation}
Assuming without loss of generality a harmonic dependence of the wave function on the angular coordinate $\chi$,
the stationary Schr\"odinger equation is
\begin{equation}
H(e^{iS\chi}\Phi)=Ee^{iS\chi}\Phi
\end{equation}
for a function $\Phi: \MAH\rightarrow \CC$.  This is equivalent to 
\begin{equation}
\label{AHeigen}
-\frac{1}{abcf}\frac{\partial}{\partial r}\left(\frac{abc}{f}\frac{\partial\Phi}{\partial r}\right)+\left(\frac{\xi_{1}^{2}}{a^{2}}+\frac{\xi_{2}^{2}}{b^{2}}+\frac{\xi_{3}^{2}}{c^{2}}\right)\Phi=\epsilon\Phi,
\end{equation}
where $\epsilon=\frac{4\pi E}{\hbar^{2}}-\frac{S^2}{4}$. The quantum number $S$ is necessarily an integer and characterises the total electric charge of the quantum state \cite{GM}.

We will now derive and study spectral problems for functions on the half-line which can be obtained from \eqref{AHeigen} by separating the  dependence of the  function $\Phi$  on  the angular coordinates and the radial coordinate $r$. For details about the separation of variables we refer the reader to \cite{GM} and \cite{S}. Here we only give enough background to help the reader appreciate the interpretation of the bound states which we will encounter in terms of magnetic monopoles.

The vector fields $\xi_1,\xi_2$ and $\xi_3$ generate rotations of a two-monopole configuration about body-fixed orthogonal axes. The sum of squares $
 \xi^2_1+\xi^2_2+\xi^2_3 $ 
commutes with the Laplace operator on $\MAH$ and represents the total angular momentum of the two-monopole configuration. As usual in quantum theory, it has eigenvalues $-j(j+1)$, for an integer $j\geq 0$. 

The operator $\xi_3$ does not commute with the Laplace operators on $\MAH$, but does commute with its asymptotic form where $a,b,$ and $c$ are replaced by $\atn,\btn$ and $\ctn$. For well-separated monopoles, $\xi_3$ generates the rotation about the line joining the monopoles and an eigenvalue $s$ of $-i\xi_3$ characterises the relative electric charge of the two monopoles. The metamorphosis of body-fixed relative angular into electric charge is one of the interesting and subtle aspects of the theory of non-abelian monopoles. Again we refer the reader to \cite{GM}  for details. Finally, note that, as a consequence of the $\ZZ_2$-division in \eqref{modmop}, the relative electric charge $s$ and the total electric charge $S$ have to have an even sum. This essentially reflects the fact that the individual electric charges (only defined asymptotically) are both integers.

To study the spectrum of \eqref{AHeigen},  we separate variables in terms of  Wigner functions $D^{j}_{sm}$ for the dependence on $SO(3)$. Referring to \cite{GM,S} for details, the conservation of the total angular momentum but not of the relative electric charge means that one may  fix $j$ but needs to consider linear combinations of Wigner functions with all  allowed values of $s$,  with coefficient functions $u_{js}$ of the radial coordinate. This leads to systems of coupled ordinary differential equations, increasing in size with $j$, whose structure is described in \cite{S}. It turns out that, because of parity considerations, only a single radial equation needs to be considered for $j=0$ (where necessarily $s=0$) and also for $(j,s)= (1,1), (2,1)$ or $(3,2)$. In the case $j=0$ there are no bound states (albeit very interesting scattering, see \cite{S}), but the other three single channels support bound states, which we now discuss.
The radial equation takes the form 
\begin{equation}
-\frac{1}{abcf}\frac{d}{dr}\left(\frac{abc}{f}\frac{du_{js}}{dr}\right)+V_{js}u_{js}=\epsilon u_{js},
\label{lmnHami}
\end{equation}
where the potentials $V_{js}$ are given in Table \ref{potentialsjs}. 
\begin{table}[H]
\centering
\begin{tabular}{l|l}
\thead{$(j,s)$} & \thead{$V_{js}$} \\
\hline
$(1,1)$ & $\frac{1}{b^{2}}+\frac{1}{c^{2}}$ \\
$(2,1)$ & $\frac{4}{a^{2}}+\frac{1}{b^{2}}+\frac{1}{c^{2}}$\\
$(3,2)$ & $\frac{4}{a^{2}}+\frac{4}{b^{2}}+\frac{4}{c^{2}}$\\
\end{tabular}
\caption{The potentials for the three single channels with bound states in terms of $a$,$b$ and $c$.} 
\label{potentialsjs}
\end{table}
 
Bound state energies for the channels $(j,s) = (1,1)$ and  $(j,s) = (2,1)$ were computed numerically in \cite{M} using a shooting method applied to  \eqref{lmnHami}. The bound states in the channel  $(j,s) = (3,2)$ were missed in \cite{M} but pointed out in \cite{BJST}, where their bound state energies were also computed using a shooting method. The results in \cite{Manton,BJST} are not rigorous. Our goal is to use our results from Sect.~2 to prove that infinitely many bounds states do indeed exist in each of these three channels, and to provide lower bounds for the eigenvalues.

In order to apply the results from Sect.~2, we need to `flatten' the radial derivative using equation \eqref{radlap}.
To do this, we first change coordinates to the proper radial distance coordinate $R$, defined in \eqref{radial}, which satisfies $dR= -fdr$.
With 
\begin{equation}
\nu=\sqrt{-abc},
\end{equation}
the radial derivative becomes
\begin{equation}
-\frac{1}{abcf}\frac{d}{dr}\frac{abc}{f}\frac{d}{dr}=-\frac{1}{\nu^{2}}\frac{d}{dR}\nu^{2}\frac{d}{dR}.
\end{equation}
Defining
\begin{equation}
\eta=\nu u
\end{equation}
and  substituting $u=\frac{\eta}{\nu}$ into \eqref{lmnHami} we obtain
\begin{equation}
\label{Rsystem}
H_{js}\eta=\epsilon\eta,
\end{equation}
where
\begin{equation}
\label{Hjsdef}
H_{js}=-\frac{d^{2}}{dR^{2}}+V^{\text{\tiny eff}}_{js},\qquad 
V^{\text{\tiny eff}}_{js}=\frac{1}{\nu}\frac{d^{2}\nu}{dR^{2}}+V_{js}.
\end{equation}
This is the promised reduction of the Atiyah-Hitchin Laplacian to a standard Schr\"odinger operator on the half-line. 

\subsection{Bound states of the Atiyah-Hitchin Laplacian}
\label{AHlapbound}

The Sturm-Liouville problem \eqref{Rsystem} has the form required to apply the result of Sect.~2. The potentials  $V^{\text{\tiny eff}}_{js}$ are analytic on $(0,\infty)$ since they are implicitly defined in terms of elliptic functions.  With $R$ related to $r$ via an integral, and $a,b,c$ being determined in terms of $r$ via the relations \eqref{rbeta}, \eqref{abcw} and \eqref{ws}, we have not been able to express $V^{\text{\tiny eff}}_{js}$ in terms of $R$ explicitly. However, we can determine the asymptotic information near $R=0$ and $R=\infty$ required to establish the existence of infinitely many bound states with the result of Sect.~2, and to give numerical estimates for the eigenvalues.

Near the core of the Atiyah-Hitchin manifold, we can use \eqref{smallR} to determine leading terms in $V^{\text{\tiny eff}}_{js}$ in the limit $R\rightarrow 0$. We find $ \nu=\sqrt{R} + \operatorname{O}\left( R \right)$
and therefore 
\begin{equation}
\frac{1}{\nu}\frac{d^{2}\nu }{dR^{2}} = -\frac{1}{4R^{2}}+ \operatorname{O}\left( 1 \right).
\end{equation}
Collecting leading terms in $V^{\text{\tiny eff}}_{js} $ for $R\rightarrow 0$, we arrive at the second column of Table \ref{asypots}.

For large $r$, we need to combine the behaviour of the Atiyah-Hitchin metric coefficients with respect to $r$ given in \eqref{larger} with the relation \eqref{Rr} between $r$ and the proper radial distance $R$ to derive the asymptotic behaviour of $V^{\text{\tiny eff}}_{js}(R)$. The basic method is to compute asymptotics with respect to $r$ and then deduce from \eqref{Rr} that, for example, 
\bee
\lim_{R\rightarrow \infty} \frac{R}{r^2} =0,
\eee
and so, by definition, 
\bee
\frac{1}{r^2} = \operatorname{o}\left(\frac 1 R\right) \quad \text{for} \quad R\rightarrow \infty.
\eee
One then  finds, for example,
 \bee
 \frac {1} {\atn^2} = \frac {1}{ r^2} + \operatorname{O}  \left( \frac {1} {r^3} \right), \qquad 
 \frac{1}{\ctn^2} = \frac 1 4 - \frac {1}{2r},
 \eee
and therefore in particular
 \bee
 \frac {1} {a^2} =\operatorname{o} \left( \frac {1} {R} \right), \quad  \frac{1}{c^2}= \frac 1 4 - \frac {1}{2R} + \operatorname{o} \left( \frac {1} {R} \right).
 \eee
Similarly,
 \begin{equation}
\nu = \sqrt{2}r -1-\frac{3}{4r} + \operatorname{O}\left( \frac {1} {r ^2} \right),   
\end{equation}
implies 
\bee
\frac{1}{\nu}\frac{d^{2}}{dR^{2}}\nu =  \operatorname{o}\left( \frac1R \right).
\eee
We collect the resulting asymptotic terms in the potentials $V^{\text{\tiny eff}}_{js} $ in the third column of Table \ref{asypots}.   
\begin{table}[H]
\centering
\begin{tabular}{l|l|l|lll}
\thead{$(j,s)$} & \thead{$V^{\text{\tiny eff}}_{js} $ near $R=0$ } & \thead{ $ V^{\text{\tiny eff}}_{js}$ for $R\rightarrow \infty$} & \thead{$c_2$}  & \thead{$C_0$}   & \thead{$C_1$}  \\ \hline 
(1,1) & $-\frac{1}{4R^{2}}+ \frac{2}{\pi^{2}} + \operatorname{O}\left( R \right)  $ &  $\frac{1}{4}-\frac{1}{2R}+ \operatorname{o}\left(\frac{1}{R} \right) $ & $-\frac{1}{4}$ &   $\frac{1}{4}$ & $-\frac{1}{2}$ \\ 
 $(2,1)$ & $\phantom{-}\frac{3}{4R^{2}}+\frac{2}{\pi^{2}}+\operatorname{O}\left( R \right) $ & $\frac{1}{4}-\frac{1}{2R}+  \operatorname{o}\left(\frac{1}{R} \right) $ & $\phantom{-} \frac{3}{4}$   &  $\frac{1}{4}$ & $-\frac{1}{2}$ \\
 $(3,2)$  & $\phantom{-}\frac{3}{4 R^{2}}+\frac{8}{\pi^{2}}+\operatorname{O}\left( R \right)$ & $1-\frac{2}{R}+\operatorname{o}\left(\frac{1}{R} \right) $ & $\phantom{-} \frac{3}{4}$  & $1$ & $-2$  \\
\end{tabular}
\caption{The asymptotic forms of the potentials $V^{\text{\tiny eff}}_{js}$ for small and large $R$.}
\label{asypots}
\end{table}

By virtue of the results of Sect.~\ref{generalresults}, we arrive at the following.

\begin{corollary}
The radial Hamiltonians $H_{js}:\mathcal{D}\longrightarrow L^2(0,\infty)$ defined in \eqref{Hjsdef}
are selfadjoint. Their essential spectrum is the segment $[C_0,\infty)$,  where $C_0$ is given in Table~\ref{asypots}. Each of these Hamiltonians has infinitely many eigenvalues below $C_0$.
\label{coro4}
\end{corollary}

\subsection{Numerical approximation of the eigenvalues} \label{AHnum}

Having established the existence of infinitely many eigenvalues, we would like to produce numerical approximations for them, in analogy with our treatment of the linearised YMH equations in Sect.~2. There we were able to exploit the asymptotic agreement between the radial  YMH Hamiltonian \eqref{pot3_24} and the radial Coulomb Hamiltonian. A natural exactly solvable approximation to the Laplace operator on the Atiyah-Hitchin manifold is provided by the Laplace operator on the (negative mass) Taub-NUT space.

 As explained in \cite{GM,S}, replacing the Atiyah-Hitchin radial functions $a,b,c$ by their Taub-NUT counterparts \eqref{TNapprox} (with $f=-b/r$ similarly replaced) in the eigenvalue equation \eqref{AHeigen}, and separating variables leads to an exactly solvable radial problem on the half-line $(0,\infty)$. The additional $U(1)$ symmetry of the Taub-NUT metric means that $\xi_3$ commutes with the Laplace operator, so that separating variables into Wigner functions of the angular coordinates and a radial function $u_{js}$, yields decoupled radial equations. Writing $u_{js}(r) = h_{js}(r)/r$, the radial equations derived in \cite{GM} are 
\begin{equation}
\left( \frac{d^2}{dr^2}-\frac{j(j+1)}{r^2}-\frac{s^2}{4}\left(1-\frac{4}{r}\right) + \epsilon\left(1-\frac{2}{r}\right)\right)h_{js}(r)=0.
\end{equation}
Remarkably, the singularity in Taub-NUT at $r=2$ is not visible in this radial equation, which can be solved exactly in terms of confluent hypergeometric functions. The relevant eigenvalues are  
\begin{equation}
\label{TNenergy}
\epsilon_{(s,n)} = \frac 1 2 \sqrt{n^2-s^2}\left(n-\sqrt{n^2-s^2}\right),\qquad |s| \leq j, \qquad  n=j+1,j+2 \ldots .
\end{equation}

One might think that the exact solutions of the Taub-NUT radial equation could be used as trial wavefunctions for the radial Atiyah-Hitchin equations \eqref{Rsystem}, in analogy with our use of the Coulomb wavefunctions in the YMH radial problem. However, there are theoretical and practical problems to overcome. The Atiyah-Hitchin and Taub-NUT manifolds are different manifolds (even topologically), and identifying radial coordinates on the two spaces is arbitrary. Pragmatically, one might identify, for example, the proper radial distance coordinate $R$ on the Atiyah-Hitchin space with the radial coordinate $r$ on the Taub-NUT space because they have the same range, and the Taub-NUT problem is most easily solved in terms of $r$. However, even with this choice, the numerical computation of the potential $V^{\text{\tiny eff}}_{js}$ in \eqref{Rsystem} as a function of $R$ with control over numerical errors is very difficult because it involves, amongst others, the inversion of the  elliptic function arising in  \eqref{rbeta}. We were therefore not able to construct useful trial wavefunctions for the Atiyah-Hitchin Laplacian from the Taub-NUT eigenfunctions by following this idea. Instead we consider a more pedestrian approach.

\begin{table}[H]
\centering
\caption{The computed eigenvalues of $H_{js}$ for three channels.}
\label{tab:evalsAHpots}
\centering
\begin{tabular}{c|c|c}
\thead{$(j,s)= (1,1)$} & \thead{$(j,s)=(2,1)$} & \thead{$(j,s)=(3,2)$} \\ \hline
0.23151604 & 0.24264773 & 0.92838765 \\
0.24250546 & 0.24597017 & 0.95655735 \\
0.24605425 & 0.24745446 & 0.97063593 \\
0.24898588 & 0.24836885 & 0.97876253 \\
           & 0.24942162 & 0.98390049 \\
           &            & 0.98736886 \\
\end{tabular}
\end{table}

In Table~\ref{tab:evalsAHpots} we show the  numerically computed  lowest four to six eigenvalues in each of the channels listed by means of the Matlab routine Chebfun, \cite{CF}. These calculations are expected to be more accurate than those in \cite{M,BJST}. We use the radial Atiyah-Hitchin Hamiltonian in the form \eqref{lmnHami}. The equations \eqref{SDcond} for the coefficient functions (with $f=-b/r$) are used to express derivatives of $a,b$ and $c$, in terms of $a,b$ and $c$. For each integration of \eqref{lmnHami}, the coefficient functions $a,b$ and $c$ are obtained from \eqref{SDcond}, starting with initial data $a=2h$, $b=\pi+h$ and $c=-\pi+h$ where $h=0.001$.

We have listed the corresponding Taub-NUT eigenvalues from \eqref{TNenergy} in Table~\ref{tab:TNapp}. For the channel $(j,s)=(1,1)$, the lowest four energies occur when $n=2,3,4,5$, for $(j,s)=(2,1)$ they occur when $n=3,4,5,6,7$ and for $(j,s)=(3,2)$ they occur when $n=4,5,6,7, 8,9$. Our calculations confirm the remarkable agreement between numerically computed eigenvalues for the Atiyah-Hitchin Hamiltonian and the Taub-NUT eigenvalues. This agreement was pointed out and discussed in \cite{M} for the three lowest lying eigenvalue in the channels $(j,s)=(1,1)$ and $(j,s)=(2,1)$, and also in \cite{BJST} for $(j,s)=(3,2)$.

Note that the Taub-NUT approximation is slightly below our numerically computed eigenvalues for the Atiyah-Hitchin Laplacian   for all but two (the lowest when $(j,s)= (1,1)$ and $n=2,3$).

\begin{table}[H]
\caption{The Taub-NUT approximation to the first eigenvalues from \eqref{TNenergy} and agreement with the eigenvalues from Table~\ref{tab:evalsAHpots}. The signs displayed correspond to whether the Taub-NUT eigenvalue is above (+) or below (-) the Atiyah-Hitchin eigenvalue.}
\label{tab:TNapp}
\centering
\begin{tabular}{cl|cl|cl}
\thead{$(j,s)=(1,1)$}& Agreement  & \thead{$(j,s)=(2,1)$} & Agreement & \thead{$(j,s)=(3,2)$} & Agreement \\ \hline
 0.23205081 &+ 99.77\%& 0.24264069  &- 99.99(7)\%& 0.92820323&- 99.98\% \\
 0.24264069 &+ 99.94\%& 0.24596669  &- 99.99(9)\%& 0.95643924&- 99.99\% \\ 
 0.24596669 &- 99.96\%& 0.24744871  &- 99.99(7)\%& 0.97056275 &- 99.99\% \\
 0.24744871 &- 99.38\%& 0.24823935  &- 99.95\%& 0.97871376 &- 99.99(5)\% \\
	    &     &  0.24744871     &- 99.20\%& 0.98386677      &- 99.99(7)\% \\
	    &	& &	  & 0.98733975             &- 99.99(7)\% \\
\end{tabular}

\end{table}

\section{Conclusion}

In this paper we rigorously established the existence of infinitely many eigenstates of operators  which arise in linearisations of the $SU(2)$ YMH equations in the BPS limit and  in symmetry reductions of the Laplace operator on the moduli space of two monopoles. We have also provided sharp numerical estimates of the eigenvalues in some cases. As promised in the Introduction, we would now like to look at the physical  interpretations of these eigenstates in YMH theory, with a particular emphasis on their significance for electric-magnetic duality conjectures. 

As discussed in Sect.~\ref{linearised}, suitable quantum superpositions of the  eigenstates  of the Schr\"odinger operator \eqref{pot3_24} define fluctuations around the BPS monopole of definite electric charge $n\in \ZZ$. The eigenstates we found therefore give rise to an infinite tower of Coulombic bound states in each of the dyonic sectors $(1,n)$, $n \in \ZZ$. These Coulombic bound states are covered by an essential spectrum arising from other sectors, see \cite{RS} for details.

The eigenstates of the Atiyah-Hitchin Laplacian discussed in Sects.~\ref{AHlapbound} and ~\ref{AHnum} describe quantum states of magnetic monopoles of charge $m=2$ and relative electric charge $s=1$ (with angular momentum $j=1$ or $j=2$) or $s=2$ (with angular momentum $j=3$). The total electric charge has to be  odd  when $s=1$ and even when $s=2$, but is otherwise arbitrary. We thus have two  families of dyonic sectors, with each sector containing  infinitely many Coulombic bound states:  one family labelled by 
$(2,n)$, $n $ odd, and  one by  $(2,n)$ with $n$ even. The bound states are also covered by an essential spectrum arising from other channels.

The families of charge sectors $(1,n)$, $ n\in \ZZ$ and $(2,n)$, $n$ odd, have both featured prominently in studies of S-duality since they are related by the $SL(2,\ZZ)$ action reviewed in the Introduction. The fact that both contain so-called quantum BPS states was one of the first pieces of strong evidence for S-duality. In the language of this paper, BPS quantum states are bound states with energy equal to the lower bound of the essential spectrum.

Here, we did not consider BPS states, but we can indicate briefly how they fit into our discussion, referring to \cite{dVS} where the N=4 supersymmetric theory is studied in notation similar to the one used here. 
The BPS states in the $(1,n)$ sectors  correspond to zero-energy solutions of the coupled system \eqref{valpha} which, as already pointed out,  cannot be obtained via the transformation \eqref{changevar} used here. In the $(2,n)$ sectors ($n$ odd), the BPS state is the famous  Sen form. It is a zero-energy eigenstate of the  Atiyah-Hitchin  Laplacian  acting on differential forms. As explained in \cite{dVS}, all eigenfunctions (zero-forms) of the Laplace operator are part of a supersymmetry multiplet of differential eigenforms of the Laplacian.  However, the zero-energy eigenstates are special, and the corresponding supersymmetric multiplet does not contain ordinary functions on the Atiyah-Hitchin space. As a result, we do not see them in our analysis.  However,  Fig.~\ref{spectrum} illustrates the relation of the BPS state to the essential spectrum of Laplace operator acting on forms and to the eigenvalues studied here. 

\begin{figure}[ht]
\centering
\includegraphics[width=.6\textwidth]{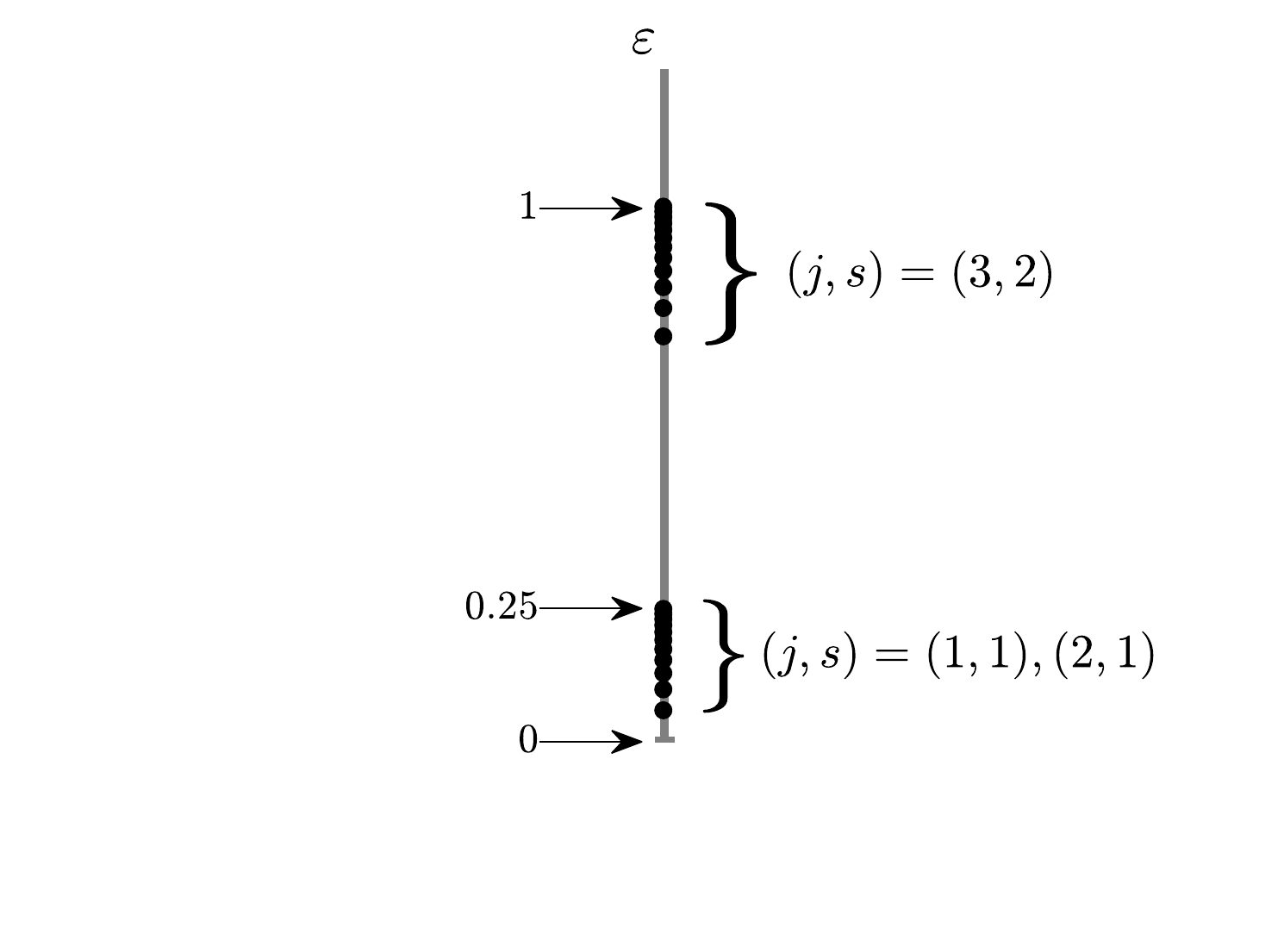}
 \caption{The qualitative nature of the spectrum of the Laplace operator (acting on differential forms) on the Atiyah-Hitchin manifold: the special eigenvalue $\varepsilon=0$ for the Sen (or BPS) state is also the lower bound of the essential spectrum. The Coulombic bound states studied here are all embedded in the essential spectrum. The spectrum of the linearised YMH operator has the same qualitative form. }
 \label{spectrum}
\end{figure}

The Coulombic families of bound states we found both  for $(1,n)$, $ n\in \ZZ$ and $(2,n)$, $n$ odd, provide further evidence for  the similarities between these two families of charge sectors, and possibly further evidence for S-duality. The latter would not require the spectra in these sectors to be equal, since it also involves a change in the Yang-Mills coupling constant. It would, however, suggest that the tower of Coulombic states found for the linearised YMH equations and for the $s=1$ channels of the Atiyah-Hitchin Laplacian represent different approximations, valid for different values of the Yang-Mills coupling constant, to the same physical system of bound states. 

The Coulombic bound states in the $(j,s)=(3,2)$ channel of the Atiyah-Hitchin Laplacian are physically the most surprising of the bound states studied here. They describe bound states of dyons with charges $(1,1)$ and $(1,-1)$ and were overlooked in \cite{M} since one might expect dyons of equal and opposite electric charges to exchange their electric charge and turn into pure monopoles. As pointed out in \cite{BJST}, this does not happen because the bosonic nature of pure monopoles does not allow them to be in a state of orbital angular momentum $j=3$. Applying S-duality to these states  leads to  surprising  predictions. A pair of dyons with charges $(1,1)$ and $(1,-1)$ is S-dual to a pair with charges $(-1,1)$ and $(1,1)$.
The Coulombic bound states we found would therefore be dual to bound states in a system consisting of a monopole and an anti-monopole, both carrying one unit of electric charge (breather states).

To end, we remark that the study of  bound states in single channels arising in the linearised YMH equations and the Atiyah-Hitchin Laplacian are merely the first steps in a full exploration of these spectral problems. Both contain interesting scattering processes, with strikingly similar  qualitative features  \cite{S,RS}. For higher angular momenta, both yield infinitely many  multi-channel problems, with both bound states and scattering processes. Our analysis suggests that all of these warrant careful further study.

\vspace{1cm}

\noindent {\bf Acknowledgements.} \,
Funding for the research reported in this paper was provided by EPSRC grants EP/K00848X/1 and EP/L504774/1.

\appendix
%\appendixpage
\section{Detailed evaluation of the stiffness matrices for the linearised YMH equations} \label{massandstiffness}

To compute the stiffness matrix in \eqref{stifma} we first split the Hamiltonian $\HYMH$ into two parts   
\begin{equation}
\HYMH=H_{1}+H_{2},
\end{equation}
where
\begin{align}
H_{1}&=-\frac{d^{2}}{dr^{2}}+\frac{2}{r^{2}}-\frac{2}{r}, \\
H_{2}&=\frac{1+\cosh^{2}(r)}{\sinh^{2}(r)}+\frac{2}{r}(1-\coth(r)).
\end{align}
Now defining   
\begin{align}
S_{k}(i,j)&=\int^{\infty}_{0}H_{k}\tilde{b_{i}}(r)\tilde{b_{j}}(r)dr, \qquad k=1,2,
\end{align}
one finds that 
$S_{1}(i,j)$ is  straightforward to calculate, but that the evaluation of  $S_{2}(i,j)$ is rather involved. To  organise  it, 
we write  the basis functions  \eqref{eqforbtilde} as
\begin{equation}
\tilde{b_{i}}(r)=r\left(\sum^{i}_{p=1}d^{i+1}_{p-1}r^{p}\right)e^{-\left(\frac{1}{i+1}\right)r},
\end{equation}
and introduce the notation  $\hat{d}^{ij}_p\in \mathbb{R}$  for the coefficients of the product 
\[
      \tilde{b_{i}}(r)\tilde{b_{j}}(r)=r^2\left(\sum^{i+j}_{p=2}\hat{d}^{ij}_{p}r^{p}\right)e^{-\left(\frac{1}{i+1}+\frac{1}{j+1}\right)r}.
\] 
Then
\begin{align}
S_{2}(i,j)&=\int^{\infty}_{0} \frac{  r^{2} \left(1+\cosh^{2}(r)+\frac{2}{r}\sinh^{2}(r)-\frac{2}{r}\sinh(r)\cosh(r)\right)}{\sinh^{2}(r)}  \left(\sum^{i+j}_{p=2}\hat{d}^{ij}_{p}r^{p}\right)e^{-\left(\frac{1}{i+1}+\frac{1}{j+1}\right)r}dr\nonumber\\
&=\sum^{i+j}_{p=2}\hat{d}^{ij}_{p}\int^{\infty}_{0}\frac{r^{2}}{\sinh^{2}(r)}\left(\frac{3}{2}+\frac{e^{2r}}{4}+\frac{e^{-2r}}{4}-\frac{1}{r}+\frac{e^{-2r}}{r}\right)r^{p}e^{-\left(\frac{1}{i+1}+\frac{1}{j+1}\right)r}dr\nonumber\\
&=\sum^{i+j}_{p=2}\hat{d}^{ij}_{p}\left(\frac{3}{2}\mathscr{L}\left(\frac{r^{2+p}}{\sinh^{2}(r)}\right)\left(\frac{1}{i+1}+\frac{1}{j+1}\right)+\frac{1}{4}\mathscr{L}\left(\frac{r^{2+p}}{\sinh^{2}(r)}\right)\left(\frac{1}{i+1}+\frac{1}{j+1}-2\right)\right.\nonumber\\ 
&+\left. \frac{1}{4} \mathscr{L}\left(\frac{r^{2+p}}{\sinh^{2}(r)}\right)\left(\frac{1}{i+1}+\frac{1}{j+1}+2\right)-2
\mathscr{L}\left(\frac{r^{1+p}}{\sinh (r)}\right)
\left(\frac{1}{i+1}+\frac{1}{j+1}+1\right)\right),
\label{eq:S1matrix}
\end{align}
where $\mathscr{L}(f(t))(s)$ is the Laplace transform of $f(t)$ and we  used  the exponential definitions of the hyperbolic functions. Using \cite[25.11.25]{NIST} we see that
\begin{align}
\mathscr{L}\left(\frac{r}{\sinh (r)}\right)(s)&=\int^{\infty}_{0}\frac{2e^{-s r}r}{e^{r}-e^{-r}}dr\nonumber\\
&=\frac{1}{2}\int^{\infty}_{0}\frac{e^{-\frac{s x}{2}}x}{e^{\frac{x}{2}}-e^{-\frac{x}{2}}}dx\nonumber\\
&=\frac{1}{2}\int^{\infty}_{0}\frac{e^{-\left(\frac{s+1}{2}\right)x}x}{1-e^{-x}}dx\nonumber\\
&=\frac{1}{2}\Gamma(2)\xi\left(2,\frac{s+1}{2}\right),\nonumber
\label{rsinhGamma}
\end{align}
where $\Gamma(t)$ is the Gamma function , $\xi(q,w)$ is the Hurwitz zeta function and we have made the substitution $2r=x$. This evaluates the last term in \eqref{eq:S1matrix}. To simplify the other terms we start by defining the function
\begin{equation}
K(a,s)=\mathscr{L}\left(\frac{r^{a}}{\sinh^{2}(r)}\right)(s),\hspace{3mm}a=2,3,\ldots .
\end{equation}
Taking the derivative of $K(2,s)$ with respect to $s$ we have
\begin{align}
\partial_{s}K(2,s)&=\partial_{s}\mathscr{L}\left(\frac{r^{2}}{\sinh^{2}(r)}\right)(s)\nonumber\\
&=\int^{\infty}_{0}\frac{r^{2}}{\sinh^{2}(r)}\partial_{s}e^{-sr}dr\nonumber\\
&=-\int^{\infty}_{0}\frac{r^{3}}{\sinh^{2}(r)}e^{-sr}dr\nonumber \\
&=-K(3,s).
\end{align}
By induction we have the relation
\begin{equation}
K(a+1,s)=(-\partial_{s})^{(a-1)}K(2,s),\hspace{3mm}a=2,3,... .
\label{Kinducrel}
\end{equation}
Using \cite[25.11.25, 25.11.12]{NIST} we can write
\begin{align}
K(2,s)&=\mathscr{L}\left(\frac{r^{2}}{\sinh^{2}(r)}\right)(s)\nonumber\\
&=\int^{\infty}_{0}\frac{4r^{2}e^{-sr}}{(e^{r}-e^{-r})^{2}}dr\nonumber\\
&=\frac{1}{2}\int^{\infty}_{0}\frac{x^{2}e^{-\left(\frac{s+2}{2}\right)x}}{(1-e^{-x})^{2}}dx\nonumber\\
&=\frac{1}{2}\int^{\infty}_{0}e^{-\frac{sx}{2}}x^{2}\frac{d}{dx}\left(\frac{-1}{1-e^{-x}}\right)dx\nonumber\\
&=\int^{\infty}_{0}\frac{xe^{-\frac{sx}{2}}}{1-e^{-x}}dx-\frac{s}{4}\int^{\infty}_{0}\frac{x^{2}e^{-\frac{sx}{2}}}{1-e^{-x}}dx\nonumber\\
&=\Gamma(2)\xi\left(2,\frac{s}{2}\right)-\frac{2}{4}\Gamma(3)\xi\left(3,\frac{s}{2}\right)\nonumber\\
&=\Phi'\left(\frac{s}{2}\right)+\frac{s}{4}\Phi''\left(\frac{s}{2}\right),
\label{K2sdigamma}
\end{align}
where $\Phi(z)$ is the Digamma function and we have made the substitution $2r=x$. Substituting \eqref{K2sdigamma} into \eqref{Kinducrel} gives
\begin{equation}
K(a+1,s)=(-\partial_{s})^{(a+1)}\left(\Phi'\left(\frac{s}{2}\right)+\frac{s}{4}\Phi''\left(\frac{s}{2}\right)\right).
\label{Kinducdigamma}
\end{equation}
This can now be used to evaluate the expression for $S_{2}(i,j)$ in \eqref{eq:S1matrix} in terms of the Digamma function.
%Using \eqref{K2sdigamma} and \eqref{Kinducdigamma}, along with the fact that $\Phi'\left(\frac{s}{2}\right)+\frac{s}{4}\Phi''\left(\frac{s}{2}\right)$ is the solution of \[(-\partial_{s})^{2}\left(s\Phi\left(\frac{s}{2}\right)-s\Phi(1)-s\right)=0,\] in \eqref{eq:S1matrix} simplifies the calculation even further. 

%\bibliographystyle{plain}
%\bibliography{BSSWbib}

\end{document}